\documentclass[reqno]{amsart}

\usepackage{amsaddr}
\usepackage{amsmath}
\usepackage{amssymb}
\usepackage{amsthm}
\usepackage{bm}
\usepackage{color}
\usepackage{eucal}
\usepackage{fullpage}
\usepackage{graphicx}
\usepackage{hhline}
\usepackage{lineno}
\usepackage[sc]{mathpazo}
\usepackage{mathrsfs}
\usepackage{mathtools}
\usepackage[numbers,sort&compress,square]{natbib}
\usepackage{setspace}
\usepackage{subfig}
\usepackage{xr}
\usepackage[all,cmtip]{xy}

\newcommand*\patchAmsMathEnvironmentForLineno[1]{%
	\expandafter\let\csname old#1\expandafter\endcsname\csname #1\endcsname
	\expandafter\let\csname oldend#1\expandafter\endcsname\csname end#1\endcsname
	\renewenvironment{#1}%
	{\linenomath\csname old#1\endcsname}%
	{\csname oldend#1\endcsname\endlinenomath}}%
\newcommand*\patchBothAmsMathEnvironmentsForLineno[1]{%
	\patchAmsMathEnvironmentForLineno{#1}%
	\patchAmsMathEnvironmentForLineno{#1*}}%
\AtBeginDocument{%
	\patchBothAmsMathEnvironmentsForLineno{equation}%
	\patchBothAmsMathEnvironmentsForLineno{align}%
	\patchBothAmsMathEnvironmentsForLineno{flalign}%
	\patchBothAmsMathEnvironmentsForLineno{alignat}%
	\patchBothAmsMathEnvironmentsForLineno{gather}%
	\patchBothAmsMathEnvironmentsForLineno{multline}%
}

\theoremstyle{definition}

\newtheorem{definition}{Definition}
\newtheorem{example}{Example}
\newtheorem{lemma}{Lemma}
\newtheorem{proposition}{Proposition}
\newtheorem{remark}{Remark}
\newtheorem{theorem}{Theorem}

\title{Stationary frequencies and mixing times for neutral drift processes with spatial structure}
\author{Alex McAvoy$^{1}$, Ben Adlam$^{1,2}$, Benjamin Allen$^{1,3}$, Martin A. Nowak$^{1,4,5}$}

\address{\small$^{1}$Program for Evolutionary Dynamics, Harvard University, Cambridge, MA 02138 \\ $^{2}$School of Engineering and Applied Sciences, Harvard University, Cambridge, MA 02138 \\$^{3}$Department of Mathematics, Emmanuel College, Boston, MA 02115 \\ $^{4}$Department of Mathematics, Harvard University, Cambridge, MA 02138 \\ $^{5}$Department of Organismic and Evolutionary Biology, Harvard University, Cambridge, MA 02138}

\begin{document}

\allowdisplaybreaks

\maketitle

\begin{abstract}
	We study a general setting of neutral evolution in which the population is of finite, constant size and can have spatial structure. Mutation leads to different genetic types (``traits"), which can be discrete or continuous. Under minimal assumptions, we show that the marginal trait distributions of the evolutionary process, which specify the probability that any given individual has a certain trait, all converge to the stationary distribution of the mutation process. In particular, the stationary frequencies of traits in the population are independent of its size, spatial structure, and evolutionary update rule, and these frequencies can be calculated by evaluating a simple stochastic process describing a population of size one (i.e. the mutation process itself). We conclude by analyzing mixing times, which characterize rates of convergence of the mutation process along the lineages, in terms of demographic variables of the evolutionary process.
\end{abstract}

\section{Introduction}
At the heart of evolutionary theory lies the question of how individual-level properties affect the long-run composition of a population. Microscopic quantities including mutation rates, selective differences between competing types, and number of interaction partners can have profound effects on the evolutionary success of genetic types (``traits") \citep{wright:G:1931,maruyama:TPB:1970,maruyama:GR:1970,maruyama:TPB:1974,nagylaki:JMB:1980,nagylaki:S:1992,lieberman:Nature:2005,ohtsuki:Nature:2006,tarnita:PNAS:2009,traulsen:PNAS:2009,loewe:PTRSB:2010,hauert:JTB:2012}, but the precise relationship between local and global properties is often difficult to quantify. Neutral drift, which involves no selective differences between the traits, provides a natural setting in which to study this relationship \citep{kimura:CUP:1983,lynch:E:1986,ochman:MBE:2003,nei:MBE:2005,bloom:BD:2007}. Neutral population dynamics are also important to the study of populations with selection because the latter can be viewed as a perturbation of the former provided selection is sufficiently weak \citep{antal:PNAS:2009,antal:JTB:2009,tarnita:JTB:2009,nowak:PTRSB:2009,tarnita:PNAS:2011,gokhale:JTB:2011,wu:PLoSCB:2013,wu:Games:2013,allen:JMB:2014,catalan:CNSNS:2015,zhang:SR:2016}. The purpose of this article is to study the long-term distribution of traits arising along the lineages of a neutral population, as a function of the population's size, structure, evolutionary update rule, and underlying mutation process.

To motivate the discussion, let us first consider a population of size $N$ of haploid organisms reproducing asexually. Each organism has a type, which might indicate a particular trait. We assume that all types are equivalent from a reproductive standpoint. Let $S$ be the set of all possible types, which we initially assume is finite (or at least discrete). In the Moran process \citep{moran:MPCPS:1958}, which is a simple example of a finite-population evolutionary process, an individual is first chosen to reproduce. The offspring then replaces another individual chosen for death. When an individual of type $s$ reproduces in the population, mutation is possible, and the offspring acquires type $s'$ with probability $M_{s,s'}\in\left[0,1\right]$. This procedure repeats ad infinitum, and the resulting process can be used as a model to study evolution's effects on the long-run composition of the population.

Mutations can be represented by a stochastic matrix, $M$, which gives rise to a Markov chain on the set of possible types, $S$. The resulting stochastic process is called ``mutation process" \citep[see][]{taylor:EJP:2007}. In the mutation process, a transition from $s$ to $s'$ occurs with probability $M_{s,s'}$. When evolution is neutral, there is a meaningful comparison between the evolutionary process and the mutation process since then neither one involves different types reproducing at different rates. In effect, the mutation process is just an evolutionary process in a population of size $N=1$. A natural question to ask is whether population size, structure, or update rule affect average trait frequency relative to the mutation process. That is, is the average frequency of type $s$ in a population of any size, $N>1$, different from that of a population of size $N=1$?

Suppose that there are $n$ possible types and that a mutation involves switching to a new type uniformly at random. Specifically, with probability $u$, a mutation to a random type occurs, and with probability $1-u$, the offspring acquires the parental type. This kind of mutation, which is common in evolutionary game theory, is symmetric in the sense that it acts on all competing types in the same way. With mutations of this form, all types are all equally abundant in the stationary distribution \citep{tarnita:JTB:2009,antal:JTB:2009}. Stated differently, in a neutral evolutionary process with symmetric mutations among $n$ types in a population of size $N$, the average frequency of each type ($1/n$) is the same as one would find in a population of size $N=1$. We extend this result to any mutation process, showing that the average frequency of a type in a population (actually, even along any lineage) of any size is the same as that in a population of size $N=1$. Consequently, dispersal patterns and population size have no effect on the average frequencies of the traits.

In some ways, this result is not surprising, and special cases have indeed been known for some time \citep{birky:PNAS:1988}. When there are no selective differences between the traits, there is evidently nothing in the population that drives one trait to a higher average frequency in a population than in the original mutation process. However, this reasoning does not constitute a proof. Given the increasing interest in using neutral frequency as a baseline measure to understand strategy selection, we provide a general proof here. An interesting consequence of the proof is that it allows one to quantify the rate of convergence of the trait distribution in terms of simple demographic variables of the process. This rate, which we characterize in terms of mixing times, exhibits much more interesting behavior than the limit itself (which is not affected by demographic components).

Our results apply to not only any population structure but also to any trait space, $S$. In other words, a player's trait could be an element of a finite set, a denumerable (countably infinite) set, or even an uncountably infinite set. Finite sets, even those consisting of just two types, have traditionally been the focus of many models in evolutionary game theory \citep{sigmund:PUP:2010}. Other frameworks, such as adaptive dynamics \citep{nowak:AAM:1990,geritz:EE:1998,dieckmann:Nature:1999,killingback:PRSB:1999,killingback:AN:2002,dieckmann:TPB:2006,imhof:PRSB:2009,doebeli:PUP:2011} and even some within evolutionary game theory \citep{nowak:TPB:1990,nowak:Nature:1993,roberts:Nature:1998,wahl:JTB:1999a,wahl:JTB:1999b,oechssler:ET:2001,vanveelen:ET:2008,cleveland:NA:2013,cheung:JET:2014,cheung:GEB:2016}, allow for continuous trait spaces. Accordingly, we make no assumption that there are only finitely many traits.

\section{Neutral evolution}
Before treating stationary trait frequencies in neutral evolution, we first turn to a generic way of describing neutral evolution itself. \citet{allen:JMB:2014} provide a method for dealing with evolution in finite populations of fixed size and structure. This general framework applies also to processes with selective differences between the types, but here we recall just the portion that applies to neutral processes. For simplicity, we introduce this framework in the context of haploid, asexually-reproducing populations, but we also discuss how it can be extended to non-haploid populations with sexual reproduction.

The atom of the mathematical framework we consider is a ``replacement event" \citep{allen:JMB:2014}. A replacement event is a pair, $\left(R,\alpha\right)$, where $R\subseteq\left\{1,\dots ,N\right\}$ is the set of individuals (or locations) to be replaced, and $\alpha :R\rightarrow\left\{1,\dots ,N\right\}$ is a map where $\alpha\left(i\right)$ indicates that individual $i$ is replaced by the offspring of individual $\alpha\left(i\right)$. A neutral evolutionary process may then be described by a probability distribution over all possible replacement events, $\left\{p_{\left(R,\alpha\right)}\right\}_{\left(R,\alpha\right)}$. The probabilities $\left\{p_{\left(R,\alpha\right)}\right\}_{\left(R,\alpha\right)}$ encode both the population structure and the demographic update rule (Fig.~\ref{fig:populationStructure}). This process is neutral because the distribution over replacement events is the same at each time step, regardless of the state \cite{allen:PLOSCB:2015}. The population structure and update rule are fixed, in that the replacement events at different points in time are chosen independently from the same distribution. While this framework describes a broad variety of neutral evolutionary processes, there is one further ``unity" assumption that must be made, which we recall in slightly modified form from \citep{allen:JMB:2014}:
\begin{enumerate}
	
	\item[(unity)] for each $i=1,\ldots,N$, there exists a sequence of replacement events, $\left\{\left(R_{k},\alpha_{k}\right)\right\}_{k=1}^{\ell}$, such that \textit{(i)} $p_{\left(R_{k},\alpha_{k}\right)}>0$ for each $k=1,\dots ,\ell$ and \textit{(ii)} for each $j$, we have $\tilde{\alpha}_{1}\circ\cdots\circ\tilde{\alpha}_{\ell}\left(j\right) =i$, where the mapping $\tilde{\alpha}_{k}:\left\{1,\dots ,N\right\}\rightarrow\left\{1, \dots ,N\right\}$ is defined by $\tilde{\alpha}_{k}\left(m\right) =\alpha_{k}\left(m\right)$ for $m\in R_{k}$ and $\tilde{\alpha}_{k}\left(m\right)=m$ for $m\not\in R_{k}$.
\end{enumerate}

\begin{figure}
	\centering
	\includegraphics[width=0.8\textwidth]{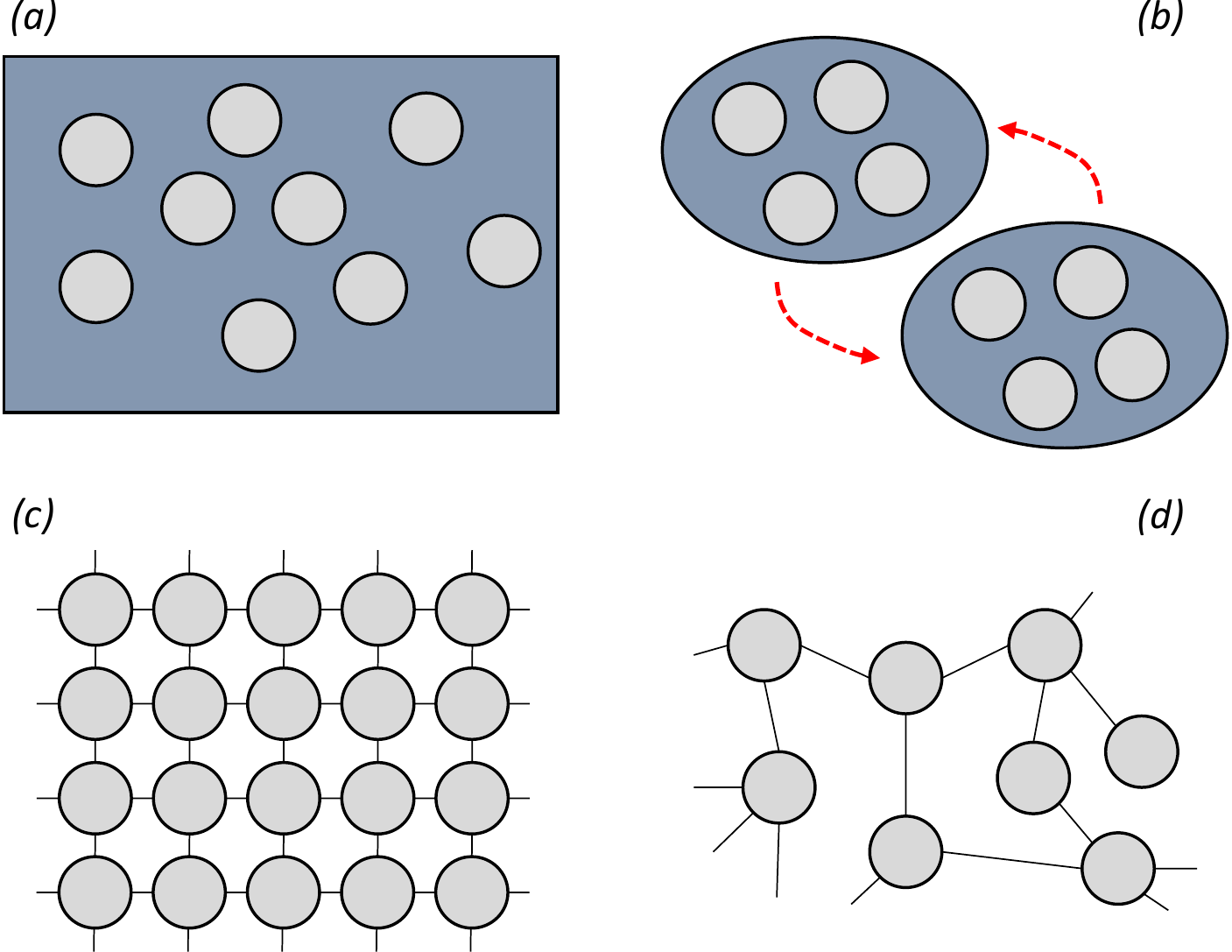}
	\caption{Four examples of evolving populations, ordered by increasing complexity of the spatial structure. In \textit{(a)}, the population is unstructured. \textit{(b)} depicts a subdivided (or ``deme-structured") population that consists of unstructured subpopulations (blue) with migration between them (red arrows). \textit{(c)} shows a regular grid (square lattice) in which all players have exactly four neighbors. \textit{(d)} illustrates a heterogeneous graph-structured population. Each individual resides on the vertex of a graph, and links indicate who is a neighbor of whom. The number of neighbors can vary from individual to individual, which results in structural asymmetries. In general, if the graph indicates an offspring-dispersal structure, then $p_{\left(R,\alpha\right)}>0$ only if $\alpha\left(i\right)$ is a neighbor of $i$ whenever $i\in R$.\label{fig:populationStructure}}
\end{figure}

This assumption formalizes the notion that the population evolves as a coherent unit, and that every individual can produce a lineage that takes over the entire population. Importantly, it does not imply that the population structure is trivial. Many interesting population structures, including heterogeneous graphs, sets, and subdivided populations with migration, satisfy the unity condition; we refer the reader to \citep{allen:JMB:2014} for further examples. In Example~\ref{ex:lineGraph} below, we give an example of a population that does not satisfy the unity condition.

At each point in time, every individual, $i=1,\ldots,N$, has a genetic type, $s_{i}\in S$. For the moment, we assume that $S$ is finite, but in the next section this assumption is relaxed. The current population state is given by the vector $\mathbf{s}=\left(s_{1},\dots ,s_{N}\right)\in S^{N}$. Let $M$ be the transition matrix for a mutation process on $S$. Transitions from a given population state $\mathbf{s} \in S^{N}$occur as follows: First a replacement event $\left(R,\alpha\right)$ is chosen from the distribution $\left\{p_{\left(R,\alpha\right)}\right\}_{\left(R,\alpha\right)}$. The probability that individual $i\in R$ is of type $s$ after this update is $M_{s_{\alpha\left(i\right)},s}$. If $i\not\in R$, the probability that individual $i$ is of type $s$ is $\delta_{s_{i},s}$ (that is, individuals who are not replaced retain their type). This process defines a Markov chain on $S^{N}$, which we denote by $\left\{\mathbf{S}\left(T\right)\right\}_{T=0}^{\infty}=\left\{\left(S_{1}\left(T\right) ,\dots ,S_{N}\left(T\right)\right)\right\}_{T=0}^{\infty}$. Note that we use $t$ to indicate time in the underlying mutation process, $\left\{X\left(t\right)\right\}_{t=0}^{\infty}$, and $T$ to denote time (i.e. number of update steps) in the evolutionary process, $\left\{\mathbf{S}\left(T\right)\right\}_{T=0}^{\infty}$; both $t$ and $T$ are discrete. The number of updates in the mutation process along a lineage is bounded from above by the number of update steps in the overall process, but in general these two measures of time need not necessarily coincide.

The distribution $\left\{p_{\left(R,\alpha\right)}\right\}_{\left(R,\alpha\right)}$ gives rise to several useful demographic variables \citep{allen:JMB:2014}, namely
\begin{align}\label{eq:demographic}
	e_{ij} \coloneqq \sum_{\substack{\left(R,\alpha\right) \\ j\in R,\ \alpha\left(j\right) =i}} p_{\left(R,\alpha\right)} ; \qquad b_{i} \coloneqq \sum_{j=1}^{N}e_{ij} ; \qquad d_{i} \coloneqq \sum_{j=1}^{N}e_{ji} .
\end{align}
$e_{ij}$ is the probability that $i$ transmits its offspring to $j$; $b_{i}$ is the birth rate of $i$; and $d_{i}$ is the death rate of $i$.

The way in which mutations are incorporated is essentially the same as they are in \citep{allen:JMB:2014}, except that here $M$ and $S$ can be arbitrary and mutations need not be symmetric with respect to the traits. 

\begin{theorem}\label{thm:ergodic}
	If the unity condition holds, then $\left\{\mathbf{S}\left(T\right)\right\}_{T=0}^{\infty}$ is ergodic whenever $\left\{X\left(t\right)\right\}_{t=0}^{\infty}$ is ergodic.
\end{theorem}
\begin{proof}
	Consider two states, $\mathbf{s},\mathbf{s}'\in S^{N}$. If $\left\{X\left(t\right)\right\}_{t=0}^{\infty}$ is ergodic, then there exists $m_{0}$ such that $\left(M^{m}\right)_{r,s}>0$ whenever $r,s\in S$ and $m\geqslant m_{0}$. If $m_{0}=1$, then $T$ is positive. Otherwise, by the unity condition, we can find an ordered sequence, $\left(R_{1},\alpha_{1}\right) ,\dots ,\left(R_{m_{0}},\alpha_{m_{0}}\right)$ with $p_{\left(R_{k},\alpha_{k}\right)}>0$ for each $k$, together with a collection $\left(i_{0},i_{1},\dots ,i_{m_{0}}\right)$ such that $\alpha_{k}\left(i_{k-1}\right) =i_{k}$ for each $k=1,\dots ,m_{0}$. After starting in state $\mathbf{s}$, and given this sequence of replacement events, the probability that player $i_{0}$ has type $s$ is $M_{s_{i_{0}},s}^{m_{0}}>0$. By the unity condition, $i_{0}$ can propagate its offspring to all other nodes in a finite number of steps; together with the fact that $M^{m}$ is positive whenever $m\geqslant m_{0}$, we see that there is a positive probability of reaching state $\mathbf{s}'$. Thus, $\left\{\mathbf{S}\left(T\right)\right\}_{T=0}^{\infty}$ is irreducible. Aperiodicity of $\left\{\mathbf{S}\left(T\right)\right\}_{T=0}^{\infty}$ follows from essentially the same argument because, after a similarly chosen sequence of replacements, the probability of staying in the state in which the chain started is positive.
\end{proof}

If the unity condition is not satisfied, then Theorem~\ref{thm:ergodic} need not necessarily hold, as the following example illustrates:
\begin{example}[Line graph]\label{ex:lineGraph}
	As an example of a process that does not satisfy the unity condition, consider a population arranged on a line. If reproduction and replacement flow in only one direction, then player ``$1$," i.e. the player at the ``beginning" of the line, is never replaced. Therefore, for each $k\in S$, there exists a stationary distribution, $\mu^{\left(k\right)}$, for $\left\{\mathbf{S}\left(T\right)\right\}_{T=0}^{\infty}$ with $\mu^{\left(k\right)}\left(s_{1}=k\right) =1$. In particular, there is more than one stationary distribution even when $\left\{X\left(t\right)\right\}_{t=0}^{\infty}$ is ergodic, so Theorem~\ref{thm:ergodic} does not hold in this case.
\end{example}

When the unity condition holds and $\left\{X\left(t\right)\right\}_{t=0}^{\infty}$ is ergodic, Theorem~\ref{thm:ergodic} implies that $\left\{\mathbf{S}\left(T\right)\right\}_{T=0}^{\infty}$ has a unique stationary distribution, $\mu$. The next section establishes the simple fact that, if $\pi$ is the stationary distribution of the mutation process, $M$, then $\mu\left(s_{i}=k\right) =\pi_{k}$ for every $k\in S$. When $\left\{\mathbf{S}\left(T\right)\right\}_{T=0}^{\infty}$ has more than one stationary distribution, it need not be the case that every such distribution satisfies $\mu\left(s_{i}=k\right) =\pi_{k}$ for every $k\in S$, which we illustrate using an irreducible mutation chain that is not ergodic:
\begin{example}[Mutation process is irreducible but not ergodic]\label{ex:irreducibleNotErgodic}
	Suppose that $S=\left\{1,2\right\}$ and let $M$ be the matrix
	\begin{align}
		M &\coloneqq 
		\begin{pmatrix}
			0 & 1 \\
			1 & 0
		\end{pmatrix} .
	\end{align}
	$\left\{X\left(t\right)\right\}_{t=0}^{\infty}$ is irreducible but periodic, with period $2$, and its (unique) stationary distribution is $\left(1/2,1/2\right)$. Consider a birth-death process on a star graph. Each player is chosen uniformly-at-random to reproduce. If a peripheral player reproduces, their offspring is subjected to the mutation operator and propagated to the central node. If the player at the central node reproduces, then this player propagates an offspring to every peripheral location (and, again, each offspring is subjected to the mutation operator). It is easy to see that the state with type $1$ at the central node and type $2$ at all of the peripheral nodes is stationary (see Fig.~\ref{fig:stars}). For this stationary distribution, $\frac{1}{N}\sum_{i=1}^{N}\mu\left(s_{i}=1\right) =1/N$, which is not equal to $\pi_{1}=1/2$ when $N>2$.
\end{example}

\begin{figure}
	\centering
	\includegraphics[width=0.8\textwidth]{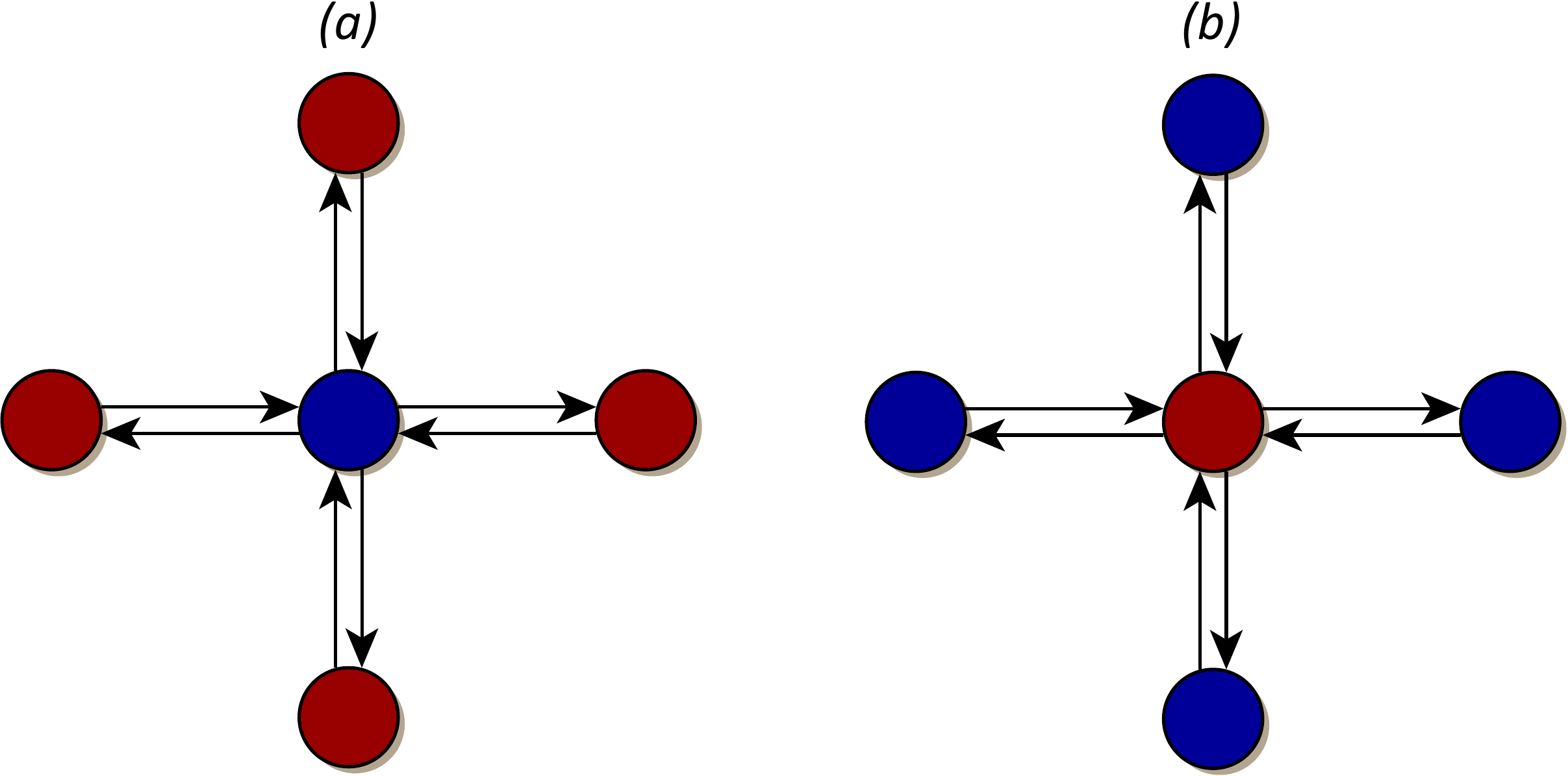}
	\caption{A birth-death process on a star graph of size $N=5$. When a player reproduces, the offspring has the opposite type of the parent (i.e. if the parent is blue, the child is red; if the parent is red, the child is blue). When a peripheral individual reproduces, the central individual dies and is replaced by the offspring. If the central individual reproduces, then all four peripheral players die and are replaced by offspring of the central player. \textit{(a)} and \textit{(b)} both give stationary states, which are unaffected by the evolutionary process. The frequency of blue is $1/5$ in \textit{(a)} and $4/5$ in \textit{(b)}, while the frequency of blue in the stationary distribution of the mutation process (i.e. in a population of size $N=1$) is $1/2$. Therefore, even though $\left\{X\left(t\right)\right\}_{t=0}^{\infty}$ is irreducible, the evolutionary process has multiple stationary distributions due to the periodicity of $M$, and these stationary distributions need not all exhibit the same trait frequencies as the original mutation process. However, there does exist a stationary distribution for the birth-death process with trait frequency $1/2$ for each of blue and red, namely the distribution that assigns probability $1/2$ to state \textit{(a)} and $1/2$ to state \textit{(b)}.\label{fig:stars}}
\end{figure}

The following result is the main tool we need to show that population size and structure do not influence stationary trait frequencies. The essence of this result is that for any $t\geqslant 0$, if one looks sufficiently far into the future, then every individual in the population will have a lineage of length at least $t$:
\begin{lemma}\label{lem:lineageLength}
	Let $B_{i}^{T}$ be the number of birth events after $T$ updates in the lineage leading to individual $i$ (a random variable). Then, for any $i=1,\dots ,N$ and any $t\geqslant 0$, we have $\lim_{T\rightarrow\infty} \mathbf{P}\left[B_{i}^{T}\geqslant t\right] =1$.
\end{lemma}
\begin{proof}
	The probability that $i$ is replaced in any given update step (i.e. the death rate of $i$) is $d_{i}$ (Eq.~\ref{eq:demographic}), which is the same at every update since the process is neutral. Let $d_{\ast}\coloneqq\min_{1\leqslant i\leqslant N}d_{i}$ and $d^{\ast}\coloneqq\max_{1\leqslant i\leqslant N}d_{i}$. Since an individual is replaced with probability at most $d^{\ast}$ and not replaced with probability at most $1-d_{\ast}$ in each step, the probability that there are exactly $t$ birth events by time $T$ on the lineage leading to $i$ satisfies
	\begin{align}
		\mathbf{P}\left[B_{i}^{T}=t\right] &\leqslant \binom{T}{t}\left(1-d_{\ast}\right)^{T-t}\left(d^{\ast}\right)^{t} .
	\end{align}
	Therefore, provided $d_{\ast}>0$, we see that $\lim_{T\rightarrow\infty} \mathbf{P}\left[B_{i}^{T}\geqslant t\right] =1$ for any $t\geqslant 0$.
\end{proof}

In the next section, we use this simple fact to derive the long-run trait frequencies.

\section{Marginal distributions and stationary frequencies}
Although the assumption that $S$ is finite is reasonable in many cases, there are also scenarios in which one would like to consider continuous trait spaces. To capture a general notion of a trait, we let $S$ be a measurable space, which contains as special cases finite, denumerable, and uncountably infinite trait spaces.

In what follows, we denote by $\mathbf{P}_{\nu}$ and $\mathbf{E}_{\nu}$ the distribution and expectation, respectively, of a random variable that depends on another distribution, $\nu$. Such is the case when $\nu$ is the initial distribution of a Markov chain and the random variable under consideration is the time-$t$ state of the chain. An element $s\in S$ in place of $\nu$ (e.g. $\mathbf{P}_{s}$ or $\mathbf{E}_{s}$) indicates that the starting point of the Markov chain is at $s$.

When dealing with a general trait space, we can no longer necessarily represent a mutation chain by a transition matrix. Instead, such a mutation process is described by a transition kernel. Let $\mathcal{F}\left(S\right)$ be a $\sigma$-algebra of subsets on $S$ and denote by $\Delta\left(S\right)$ the space of probability measures on $S$. A mutation process on $S$, $\left\{X\left(t\right)\right\}_{t=0}^{\infty}$, is then defined by a Markov kernel, $\kappa :S\rightarrow\Delta\left(S\right)$, where for $s\in S$ and $E\in\mathcal{F}\left(S\right)$, $\kappa\left(s,E\right)$ is the probability that the chain is in $E$ after being in state $s$. If $S$ is finite and the transition matrix for the mutation process is $M$, then $M_{s,s'}=\kappa\left(s,\left\{s'\right\}\right)$ for each $s,s'\in S$. To extend the notion of ergodicity to a Markov chain on a general state space, one needs the notion of a Harris chain, which we recall from \citep{durrett:CUP:2009}:
\begin{definition}
	A Markov chain, $\left\{X\left(t\right)\right\}_{t=0}^{\infty}$, on $S$ with kernel $\kappa$ is a Harris chain if there exist $A,B\in\mathcal{F}\left(S\right)$, $\varepsilon >0$, a function $q:A\times B\rightarrow\mathbb{R}$ with $q\left(s,s'\right)\geqslant\varepsilon$ for every $s\in A$ and $s'\in B$, and $\rho\in\Delta\left(B\right)$ such that
	\begin{enumerate}
		
		\item[\textit{(i)}] $\mathbf{P}_{s}\left[\tau_{A}<\infty\right] >0$ for every $s\in S$, where $\tau_{A}=\inf\left\{t\in\left\{0,1,2,\dots\right\}\ |\ X\left(t\right)\in A\right\}$;
		
		\item[\textit{(ii)}] $\kappa\left(s,C\right)\geqslant\int_{s'\in C}q\left(s,s'\right)\,d\rho\left(s'\right)$ for every $s\in A$ and $C\in\mathcal{F}\left(B\right)$.
		
	\end{enumerate}
\end{definition}

Furthermore, recurrence and aperiodicity for Harris chains are defined as follows \citep{huber:CRC:2016}:
\begin{definition}
	A Harris chain on $S$, $\left\{X\left(t\right)\right\}_{t=0}^{\infty}$, is recurrent if $\mathbf{P}_{s}\left[\tau_{A}<\infty\right] =1$ for every $s\in A$. A recurrent Harris chain is aperiodic if for every $s\in S$, there exists $t_{0}$ such that $\mathbf{P}_{s}\left[X\left(t\right)\in A\right] >0$ whenever $t\geqslant t_{0}$.
\end{definition}

We refer to a recurrent, aperiodic Harris chain as ergodic. The key result we need is the following, which can be found in \citep{huber:CRC:2016}: if $\kappa$ defines an ergodic Harris chain with stationary distribution $\pi\in\Delta\left(S\right)$, then
\begin{align}
	\lim_{t\rightarrow\infty} \sup_{E\in\mathcal{F}\left(S\right)} \left| \kappa^{t}\left(s,E\right) - \pi\left(E\right) \right| &= 0
\end{align}
whenever $s\in S$ satisfies $\mathbf{P}_{s}\left[ R_{A}<\infty \right] =1$, where $R_{A}=\inf\left\{t\in\left\{1,2,\dots\right\}\ |\ X\left(t\right)\in A\right\}$. In other words, the $t$-step transition kernel when starting from $s$, $\kappa^{t}\left(s,\--\right)$, converges in total variation to $\pi$ as $t\rightarrow\infty$.

Our main result with respect to marginal trait distributions can be stated succinctly as follows:
\begin{theorem}\label{thm:nuAndMu}
	$\lim_{T\rightarrow\infty}\mathbf{P}_{\mu_{0}}\left[S_{i}\left(T\right)\in E\right] =\pi\left(E\right)$ for every $i=1,\ldots ,N$; initial distribution, $\mu_{0}\in\Delta\left(S^{N}\right)$; and $E\in\mathcal{F}\left(S\right)$.
\end{theorem}
\begin{proof}
	Since $\left\{X\left(t\right)\right\}_{t=0}^{\infty}$ is ergodic, for each $s\in S$ and $\varepsilon >0$, there exists $t_{0}\geqslant 0$ such that whenever $t\geqslant t_{0}$,
	\begin{align}
		\sup_{E\in\mathcal{F}\left(S\right)} \left| \mathbf{P}_{s}\left[X\left(t\right)\in E\right] -\pi\left(E\right) \right| &< \varepsilon .
	\end{align}
	Thus, for each $E\in\mathcal{F}\left(S\right)$ and $t\geqslant t_{0}$, we have $\pi\left(E\right) -\varepsilon <\mathbf{P}_{s}\left[X\left(t\right)\in E\right] <\pi\left(E\right) +\varepsilon$, which gives
	\begin{align}
		\left(\pi\left(E\right) -\varepsilon\right)\mathbf{P}\left[B_{i}^{T}\geqslant t_{0}\right] < \sum_{t=t_{0}}^{T} \mathbf{P}_{s}\left[X\left(t\right)\in E\right] \mathbf{P}\left[B_{i}^{T}=t\right] < \left(\pi\left(E\right) +\varepsilon\right)\mathbf{P}\left[B_{i}^{T}\geqslant t_{0}\right] .
	\end{align}
	Since $\lim_{T\rightarrow\infty}\mathbf{P}\left[B_{i}^{T}\geqslant t_{0}\right] =1$ for any $t_{0}\geqslant 0$, we have
	\begin{align}
		\lim_{T\rightarrow\infty}\sum_{t=0}^{T} \mathbf{P}_{s}\left[X\left(t\right)\in E\right] \mathbf{P}\left[B_{i}^{T}=t\right] &= \sum_{t=0}^{t_{0}-1} \mathbf{P}_{s}\left[X\left(t\right)\in E\right] \lim_{T\rightarrow\infty}\mathbf{P}\left[B_{i}^{T}=t\right] \nonumber \\
		&\quad + \lim_{T\rightarrow\infty}\sum_{t=t_{0}}^{T} \mathbf{P}_{s}\left[X\left(t\right)\in E\right] \mathbf{P}\left[B_{i}^{T}=t\right] \nonumber \\
		&= \lim_{T\rightarrow\infty}\sum_{t=t_{0}}^{T} \mathbf{P}_{s}\left[X\left(t\right)\in E\right] \mathbf{P}\left[B_{i}^{T}=t\right] \nonumber \\
		&\in \left(\pi\left(E\right) -\varepsilon ,\pi\left(E\right) +\varepsilon\right) . \label{eq:limitAlongLineage}
	\end{align}
	Since Eq.~\ref{eq:limitAlongLineage} holds for every $s\in S$, and since $\varepsilon >0$ was arbitrary, it follows that $\lim_{T\rightarrow\infty}\mathbf{P}_{\mu_{0}}\left[S_{i}\left(T\right)\in E\right] =\pi\left(E\right)$ for every $\mu_{0}\in\Delta\left(S^{N}\right)$.
\end{proof}

Let $\delta_{s}$ denote the Dirac measure on $S$ centered at $s\in S$. The frequency of players whose type lies in $E\in\mathcal{F}\left(S\right)$ at time $T$ is
\begin{align}
	\mathbf{E}_{\mu_{0}}\left[ \frac{1}{N} \sum_{i=1}^{N} \delta_{S_{i}\left(T\right)}\left(E\right)\right] &= \frac{1}{N} \sum_{i=1}^{N} \mathbf{P}_{\mu_{0}}\left[S_{i}\left(T\right)\in E\right] ,
\end{align}
which approaches $\pi\left(E\right)$ as $t\rightarrow\infty$ by Theorem~\ref{thm:nuAndMu}; therefore, $\pi$ gives the long-term trait frequencies. In the setting of Theorem~\ref{thm:ergodic}, if $\left\{\mathbf{S}\left(T\right)\right\}_{T=0}^{\infty}$ has a unique stationary distribution, $\mu$, then $\mu\left(s_{i}=k\right) =\pi_{k}$ for every $i=1,\ldots,N$ and $k\in S$, i.e. the marginal distribution of each individual is exactly $\pi$.

Although our focus has been on haploid individuals, we note that Theorem \ref{thm:nuAndMu} can be extended to populations with diploid, haplodiploid, or polyploid genetics, with either sexual or asexual reproduction (or a combination of both) as long as \textit{(i)} only a single genetic locus is considered, \textit{(ii)} there is no recombination within this locus, and \textit{(iii)} the theorem is understood to apply at the level of alleles, rather than allele combinations or traits. The idea of this extension is to treat the alleles at this locus as asexual replicators, and let $i=1,\ldots, N$ index \emph{genetic sites} rather than individuals \citep{allen:arxiv:2018}. These genetic sites are distributed among individuals, so that each individual has a number of genetic sites equal to its ploidy. Each genetic site $i$ contains a single allele $s_i \in S$, where $S$ represents the set of possible alleles at this locus. During transitions, the alleles in a subset $R$ of genetic sites are replaced by (possibly mutated) copies of the alleles in other genetic sites, as determined by the chosen replacement event, $(R,\alpha)$. The probability distribution over replacement events, $\left\{p_{\left(R,\alpha\right)}\right\}_{\left(R,\alpha\right)}$, encodes all necessary information about ploidy, sexes, and mating.

Formally, this gene's-eye framework is mathematically equivalent to the framework based on haploid individuals \citep{allen:arxiv:2018}. All of our results therefore carry over to sexually-reproducing populations without any additional mathematical assumptions. In applying these results to sexually-reproducing populations, there is, however, an implicit \emph{biological} assumption that there is no recombination \emph{within the locus in question}. That is, each allele in a new offspring is a copy of a single allele in one parent, not a mixture of two (or more) alleles, which is reasonable if the locus represented by $S$ is small enough that linkage within the locus can be assumed complete. We also emphasize that Theorem \ref{thm:nuAndMu} characterizes the limiting frequencies of alleles but not of allele combinations in individuals. For example, in a diploid population with alleles $A$ and $a$, Theorem \ref{thm:nuAndMu} characterizes the limiting frequencies of $A$ and $a$, but not of $AA$, $Aa$, and $aa$. If different allele combinations correspond to different traits, Theorem \ref{thm:nuAndMu} does not characterize the trait distribution in this case. The limiting frequencies of allele combinations in individuals depend on the population structure, i.e.~on the replacement rule.

Even if the mutation process does not have a unique stationary distribution (such as when there are multiple absorbing states), the proof of Theorem~\ref{thm:nuAndMu} can still be used to derive the marginal distributions of every individual when the starting condition is monomorphic. Suppose that every individual in the population initially has trait $s\in S$. There will eventually be enough births along every lineage for the mutation process to reach its limiting distribution (provided the limit exists). Moreover, since the lineage leading to $i$ at time $T$ can start at any $j\in\left\{1,\dots ,N\right\}$ at time $0$, we know the initial condition of this chain, $s$, provided the population starts out monomorphic. We state this observation as a proposition:
\begin{proposition}\label{prop:notErgodic}
	If $s\in S$, $E\in\mathcal{F}\left(S\right)$, and $i\in\left\{1,\dots ,N\right\}$, then, provided the limits exist,
	\begin{align}
		\lim_{T\rightarrow\infty}\mathbf{P}_{\left(s,\dots ,s\right)}\left[ S_{i}\left(T\right)\in E \right] &= \lim_{t\rightarrow\infty}\mathbf{P}_{s}\left[X\left(t\right)\in E\right] .
	\end{align}
\end{proposition}

\begin{example}[Regeneration process]
	Suppose that each individual in a population is a (bacterial) cell dividing at a constant rate, and mutations occur during cell division. We denote by $A_{0}$ the ``wild type'' cells, which mutate to $A_{1}$ cells with probability $w$. $A_{1}$ cells contain the starting condition for the search process (for example, the duplicated gene). For $k>0$, mutation with probability $u$ leads from an $A_{k}$ cell to an $A_{k+1}$ cell. These steps are the forward mutations in the search process toward the new function. The search is lost with probability $v$, i.e. each $A_{k}$ cell (with $k>0$) mutates back to an $A_{0}$ cell with probability $v$. $v$ can represent the rate of deletion events, nonsense mutations, or any missense mutation that leads away from the target because then the search is essentially lost. It is natural to assume that $v>u$, meaning that at each step, it is more likely that the search is lost than that a mutation is made in the direction of the target. This mutation scheme, which was introduced by \citet{knoll:SA:2017}, is known as the ``regeneration process."
	
	Consider a Wright-Fisher process in a population of size $N$. Generations are non-overlapping, and at each time step the current generation is sampled (with replacement) from the previous generation \citep{ewens:S:2004,imhof:JMB:2006,der:TPB:2011}. In other words, every individual in the current generation is the offspring of individual $i$ in the previous generation with probability $1/N$ for $i=1,\dots ,N$. The total population size is strictly constant, and all cells have the same reproductive rate. In other words, evolution is neutral. On each birth event, the offspring mutates from the parent's type according to the regeneration process. We are interested in the stationary distribution of this stochastic process, and, in particular, the probability that a randomly-drawn cell is of type $A_{k}$. Denote by $\pi_{k}$ this probability and suppose first that $N=1$. By the definition of a stationary distribution and the mutational scheme for the regeneration process, we have the recurrence relations
	\begin{subequations}
		\begin{align}
			\pi_{0} &= \pi_{0}\left(1-w\right) + \left(1-\pi_{0}\right) v ; \\
			\pi_{1} &= \pi_{0} w + \pi_{1}\left(1-u-v\right) ; \\
			\pi_{k} &= \pi_{k-1}u + \pi_{k}\left(1-u-v\right) ; \qquad \left( 1<k<m \right) \\
			\pi_{m} &= \pi_{m-1}u + \pi_{m}\left(1-v\right) .
		\end{align}
	\end{subequations}
	The solution to this system of equations is given explicitly by
	\begin{subequations}\label{eq:stationarySolution}
		\begin{align}
			\pi_{0} &= \frac{v}{v+w} ; \\
			\pi_{k} &= \left(\frac{u}{u+v}\right)^{k-1}\left(\frac{w}{u+v}\right)\left(\frac{v}{v+w}\right) ; \qquad \left( 0<k<m \right) \\
			\pi_{m} &= \left(\frac{u}{v}\right)\left(\frac{u}{u+v}\right)^{m-2}\left(\frac{w}{u+v}\right)\left(\frac{v}{v+w}\right) .
		\end{align}
	\end{subequations}
	For simplicity, we assume that $u$ and $v$ do not depend on $A_{k}$. However, one can similarly calculate $\pi$ when $u_{k}$ is the probability of going from $A_{k}$ to $A_{k+1}$ and $v_{k}$ is the probability of going from $A_{k}$ to $A_{0}$.
	
	Determining the equilibrium frequency of $A_{k}$ in a population of any size, $N$, is actually quite simple: it is still just $\pi_{k}$ and can be calculated by considering a stochastic process that follows a single lineage. That is, the frequencies of the cell types are the same for a population of size $N=1$ as it is for a population of any size. Furthermore, the probability that individual $i$ has type $A_{k}$ is $\pi_{k}$ (which is independent of $i$). Fig.~\ref{fig:timeSeries} shows that heterogeneity induced by population structure can, however, strongly influence a type's fluctuation around its mean frequency.
	
	\begin{figure}
		\centering
		\includegraphics[width=\textwidth]{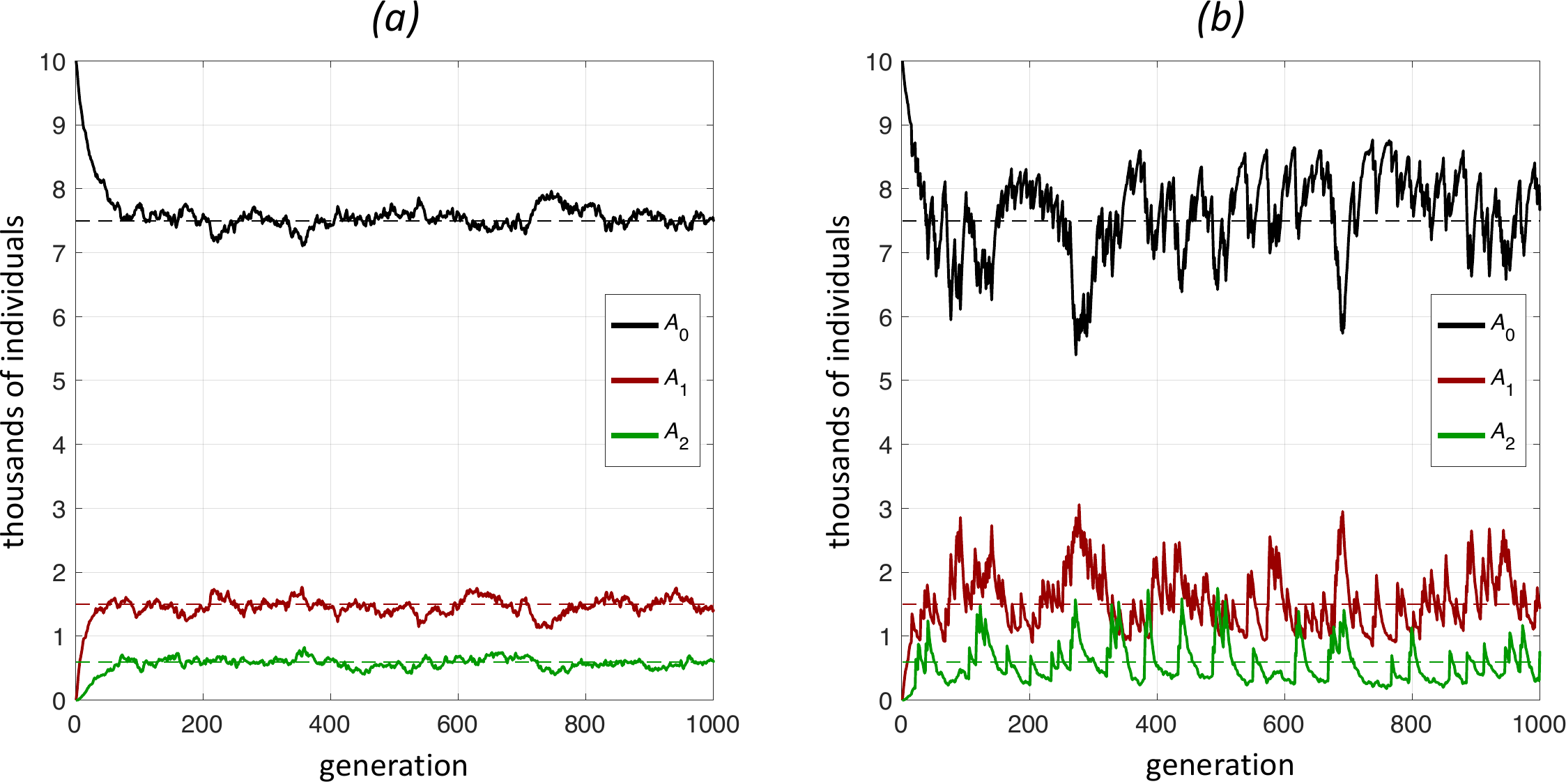}
		\caption{The frequency of three cell types in a population of size $N=10^{4}$. Mutations in both panels are governed by the regeneration process with $m=10^{3}$ and mutation rates $w=0.01$, $u=0.02$, and $v=0.03$. The mean abundance of each type in this mutation process is indicated by a dashed line. In \textit{(a)}, the population is updated according to a Wright-Fisher rule. \textit{(b)} illustrates a modified (and asymmetric) Wright-Fisher rule in which there exists a single cell in every generation that is more likely to reproduce than the others (independent of cell type). In \textit{(a)}, each offspring chooses a parent uniformly-at-random from the previous generation (i.e. each with probability $1/N=10^{-4}$). In \textit{(b)}, one marked individual is chosen as the parent with probability $0.05$ and each of the $N-1$ remaining individuals is chosen with probability $\left(1-0.05\right)/\left(N-1\right)\ll 0.05$. Despite the heterogeneity in \textit{(b)}, the mean trait abundances are the same; only the fluctuations around these means change.\label{fig:timeSeries}}
	\end{figure}
\end{example}

\begin{example}[One-dimensional Markov chain]
	Another simple example is when the mutation process, defined by transition matrix $M$, is one-dimensional. In this case, the state space is $S=\left\{0,1,\dots ,m\right\}$ for some $m\geqslant 0$, and transitions from a given state are at most one step in either direction. In other words, $M_{s,s'}=0$ whenever $\left| s-s'\right| >1$. One can easily check that the stationary distribution for $M$ is given by
	\begin{align}\label{eq:statDistOneDim}
		\pi_{s} &= \frac{\prod_{k=0}^{s-1}\frac{M_{k,k+1}}{M_{k+1,k}}}{1+\sum_{s=1}^{m}\prod_{k=0}^{s-1}\frac{M_{k,k+1}}{M_{k+1,k}}} .
	\end{align}
	Eq.~\ref{eq:statDistOneDim} gives a simple expression for the long-run probability of finding an individual in type $s$, both in the mutation process itself and in any evolving population in which mutations are neutral and governed by $M$.
\end{example}

\section{Intra-lineage mixing times}
So far, we have focused on neutral trait frequencies, which are not affected by the distribution over replacement events, $\left\{p_{\left(R,\alpha\right)}\right\}_{\left(R,\alpha\right)}$. We now turn to rate of convergence to these equilibrium frequencies along the lineages, characterized in terms of mixing times. The proof of Theorem~\ref{thm:nuAndMu} is based on the fact that, eventually, all lineages will contain sufficiently many birth events for the mutation process to mix. Here, we consider the amount of time--measured in update steps--for this mixing to occur along a given lineage. Since both the mutation processes and evolutionary update rules we consider are discrete in time, all references to time shall indicate non-negative integers.

Let $\pi\in\Delta\left(S\right)$ be the stationary distribution for the mutation process. For $t\in\left\{1,2,\dots\right\}$, let
\begin{align}
	\psi\left(t\right) \coloneqq \sup_{s\in S}\sup_{E\in\mathcal{F}\left(S\right)}\left| \mathbf{P}_{s}\left[X\left(t\right)\in E\right] - \pi\left(E\right) \right| .
\end{align}
The $\varepsilon$-mixing time of the mutation process is then $\tau\left(\varepsilon\right)\coloneqq\min\left\{ t\in\left\{1,2,\dots\right\}\ :\ \psi\left(t\right) <\varepsilon\right\}$. (We use the symbol $\psi$ to denote this distance rather than the usual $d$ in order to avoid confusion with death rates.)

At the population level, we can consider an analogue of mixing time along lineages. Recall that $B_{i}^{T}$ is the number of birth events in the lineage leading to individual $i$ after $T$ update steps. Since the marginal distributions all converge to $\pi$, the distance between $i$'s trait at time $T\in\left\{1,2,\dots\right\}$ and the stationary distribution is
\begin{align}
	\Psi_{i}\left(T\right) &= \sup_{\mathbf{s}\in S^{N}}\sup_{E\in\mathcal{F}\left(S\right)} \left| \mathbf{P}_{\mathbf{s}}\left[ S_{i}\left(T\right)\in E \right] - \pi\left(E\right) \right| .
\end{align}
The analogue of $\tau\left(\varepsilon\right)$ in this case is $\mathcal{T}_{i}\left(\varepsilon\right)\coloneqq\min\left\{ T\in\left\{1,2,\dots\right\}\ :\ \Psi_{i}\left(T\right) <\varepsilon\right\}$, which is the $\varepsilon$-mixing time of the process along the lineage leading to individual $i$. We have the trivial lower bound $\mathcal{T}_{i}\left(\varepsilon\right)\geqslant\tau\left(\varepsilon\right)$ for all $i=1,\dots ,N$.

\begin{example}[Non-overlapping generations]
	If generations do not overlap, then $\mathbf{P}\left[B_{i}^{T}=T\right] =1$. Trivially, then, we have $\Psi_{i}\left(T\right) =\psi\left(T\right)$ and $\mathcal{T}_{i}\left(\varepsilon\right) =\tau\left(\varepsilon\right)$ for every $i=1,\dots ,N$.
\end{example}

Suppose, for instance, that the death rate is constant and equal to $d\in\left(0,1\right)$, so that each individual is updated every $1/d$ update steps (on average). Since all individuals have the same probability of being replaced, $\mathcal{T}_{i}\left(\varepsilon\right)$ is independent of $i$, and we call this mixing time simply $\mathcal{T}\left(\varepsilon\right)$. If $\tau$ is the mixing time of the mutation chain, then a natural guess for $\mathcal{T}$ is simply $\tau /d$ because, on average, the lineage experiences a mutation (i.e. a step in the mutation chain) every $1/d$ updates. However, it turns out that $\tau /d$ is generally a bad approximation of $\mathcal{T}$ because it doesn't take into account enough information about the distribution of lineage length. In fact, one cannot even use $\tau /d$ to establish a general upper or lower bound on $\mathcal{T}_{i}$. Fig.~\ref{fig:lowerUpper} illustrates this comparison for death-birth updating with $d=1/N$, where $N$ is the population size.

\begin{figure}
	\centering
	\includegraphics[width=\textwidth]{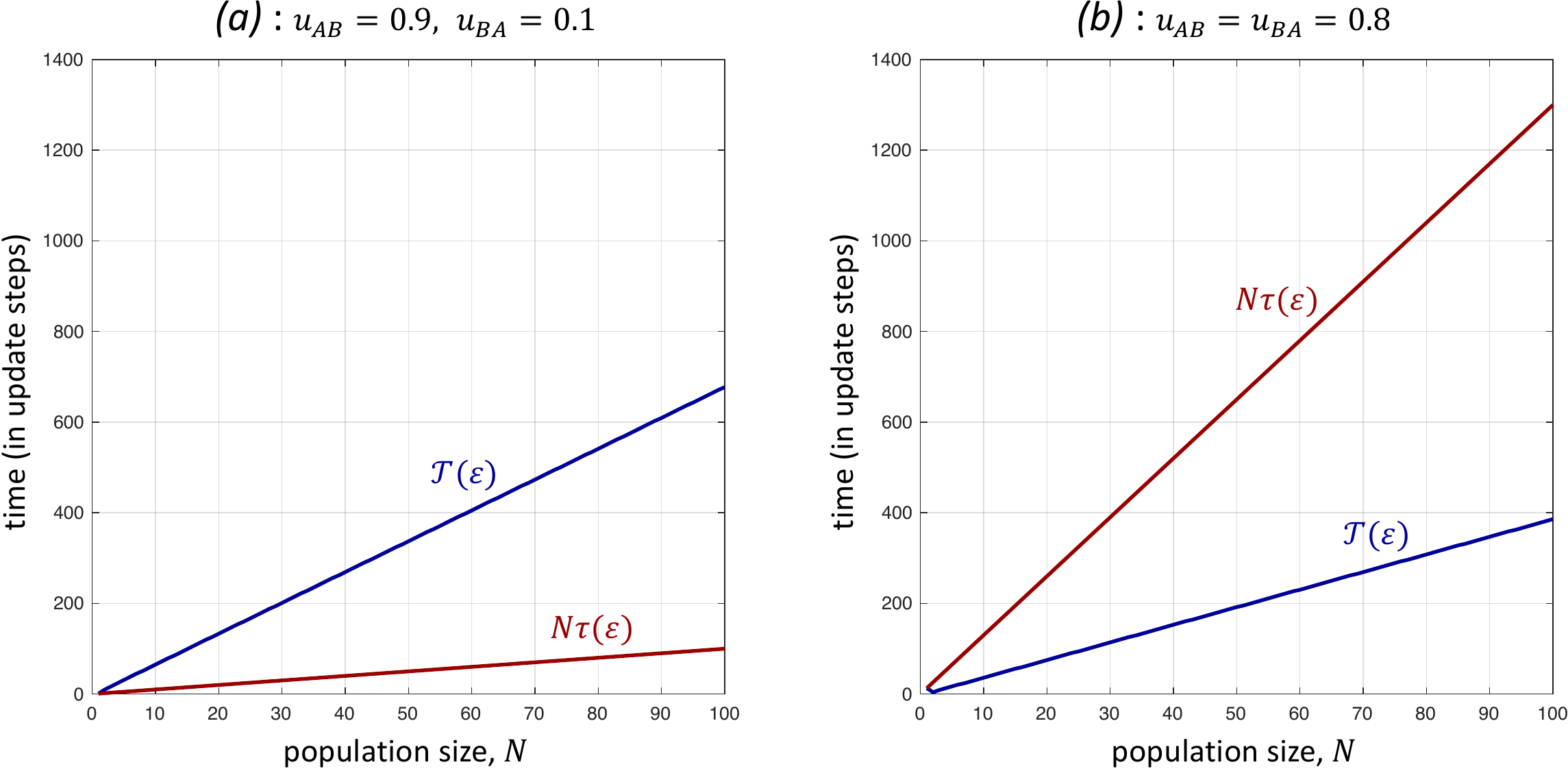}
	\caption{Intra-lineage mixing time, $\mathcal{T}\left(\varepsilon\right)$, relative to the mixing time of the underlying mutation process, $\tau\left(\varepsilon\right)$, for death-birth updating (shown here for $\varepsilon =10^{-3}$). Every individual has one of two types, $A$ or $B$, and upon reproduction $A$ mutates to $B$ with probability $u_{AB}$ and $B$ mutates to $A$ with probability $u_{BA}$. The stationary distribution of this process puts probability $\frac{u_{BA}}{u_{AB}+u_{BA}}$ on state $A$ and probability $\frac{u_{AB}}{u_{AB}+u_{BA}}$ on state $B$. Since the death rate is constant and equal to $1/N$ (regardless of the population's spatial structure), each individual is updated every $N$ steps (on average). However, $\mathcal{T}\left(\varepsilon\right)$ differs significantly from the rescaled mixing time of the mutation process, $N\tau\left(\varepsilon\right)$. In fact, $N\tau\left(\varepsilon\right)$ neither a general upper nor lower bound on $\mathcal{T}\left(\varepsilon\right)$, which can be seen by simply varying $u_{AB}$ and $u_{BA}$.\label{fig:lowerUpper}}
\end{figure}

Instead, we can bound $\Psi_{i}\left(T\right)$ by the average of $\psi\left(B_{i}^{T}\right)$, which we state as a simple lemma:
\begin{lemma}\label{lem:generalMixingBound}
	$\Psi_{i}\left(T\right)\leqslant\mathbf{E}\left[\psi\left(B_{i}^{T}\right)\right]$ for every $i=1,\dots ,N$ and $T\in\left\{1,2,\dots\right\}$.
\end{lemma}
\begin{proof}
	Fix $i=1,\dots ,N$ and $T\in\left\{1,2,\dots\right\}$. Straightforward manipulations give
	\begin{align}
		\Psi_{i}\left(T\right) &= \sup_{\mathbf{s}\in S^{N}}\sup_{E\in\mathcal{F}\left(S\right)} \left| \mathbf{P}_{\mathbf{s}}\left[ S_{i}\left(T\right)\in E \right] - \pi\left(E\right) \right| \nonumber \\
		&\leqslant \sum_{t=0}^{T} \sup_{s\in S}\sup_{E\in\mathcal{F}\left(S\right)} \left| \mathbf{P}_{s}\left[ X\left(t\right)\in E \right] - \pi\left(E\right) \right| \mathbf{P}\left[ B_{i}^{T}=t \right] \nonumber \\
		&= \mathbf{E}\left[ \sup_{s\in S}\sup_{E\in\mathcal{F}\left(S\right)} \left| \mathbf{P}_{s}\left[ X\left(B_{i}^{T}\right)\in E \right] - \pi\left(E\right) \right| \right] \nonumber \\
		&= \mathbf{E}\left[\psi\left(B_{i}^{T}\right)\right] ,
	\end{align}
	as desired.
\end{proof}

To calculate $\mathbf{E}\left[\psi\left(B_{i}^{T}\right)\right]$, we note that $\mathbf{P}\left[ B_{i}^{T}=t\right]$ satisfies the multivariate recurrence relation,
\begin{align}
	\mathbf{P}\left[ B_{i}^{T}=t \right] &= \mathbf{P}\left[ B_{i}^{T-1}=t \right]\left(1-d_{i}\right) + \sum_{j=1}^{N} \mathbf{P}\left[ B_{j}^{T-1}=t-1 \right] e_{ji} ,
\end{align}
with boundary conditions $\mathbf{P}\left[ B_{i}^{T}=0 \right] =\left(1-d_{i}\right)^{T}$ for $i=1,\dots ,N$. Explicitly, for $1<t\leqslant T$, we have
\begin{align}
	\mathbf{P}\left[ B_{i}^{T}=t \right] &= \sum_{\substack{k_{0},\dots ,k_{t}\geqslant 1 \\ k_{0}+\cdots +k_{t}=T}} \left(1-d_{i}\right)^{k_{t}-1} \sum_{j_{t-1}=1}^{N}\left(1-d_{j_{t-1}}\right)^{k_{t-1}-1}e_{j_{t-1},j_{t}}\cdots\sum_{j_{0}=1}^{N}\left(1-d_{j_{0}}\right)^{k_{0}-1}e_{j_{0},j_{1}} . \label{eq:generalBi}
\end{align}

While the general expression for $\mathbf{P}\left[ B_{i}^{T}=t\right]$ (Eq.~\ref{eq:generalBi}) can appear complicated (depending on the replacement rule, $\left\{p_{\left(R,\alpha\right)}\right\}_{\left(R,\alpha\right)}$, and the resulting demographic variables), we can further simplify the bound given by Lemma~\ref{lem:generalMixingBound} if some additional properties hold. We consider two cases in which one can be more explicit.

\subsection{Finite, ergodic mutation chains}\label{sec:finiteErgodicMixing}
If $S$ is finite and $\left\{X\left(t\right)\right\}_{t=0}^{\infty}$ is an irreducible and aperiodic, then there exist $C>0$ and $\alpha\in\left(0,1\right)$ such that $\psi\left(t\right)\leqslant C\alpha^{t}$ for every $t\geqslant 1$ \citep{levin:AMS:2008,durrett:CUP:2009}. From Lemma~\ref{lem:generalMixingBound},
\begin{align}
	\Psi_{i}\left(T\right) &\leqslant C\mathbf{E}\left[\alpha^{B_{i}^{T}}\right] . \label{eq:alphaBound}
\end{align}

\begin{example}[Constant death rate]
	If the death rate is constant, meaning there exists $d$ with $d_{i}=d$ for $i=1,\dots ,N$, then $\mathbf{P}\left[ B_{i}^{T}=t \right] =\binom{T}{t}\left(1-d\right)^{T-t}d^{t}$. From Eq.~\ref{eq:alphaBound}, we have the upper bound
	\begin{align}
		\Psi_{i}\left(T\right) &\leqslant C\left(1-d\left(1-\alpha\right)\right)^{T}
	\end{align}
	Under death-birth updating, for example, the death rate is constant and equal to $1/N$. The death rate is also clearly constant when generations are non-overlapping, and in this case $1-d\left(1-\alpha\right) =\alpha$, which gives back the same upper bound of $C\alpha^{T}$. On the other hand, when $d<1$ we have $1-d\left(1-\alpha\right) >\alpha$.
\end{example}

\subsection{Reversible mutation chains and spectral theory}
Suppose that the mutation process, $M$, is reversible with respect to $\pi$, meaning $\pi_{s}M_{s,s'}=\pi_{s'}M_{s',s}$ for every $s,s'\in S$. The matrix $\Lambda$ defined by $\Lambda_{s,s'}=\pi_{s}^{1/2}M_{s,s'}\pi_{s'}^{-1/2}$ is then symmetric with all of its eigenvalues real, $1=\lambda_{1}>\lambda_{2}\geqslant\cdots\geqslant\lambda_{\left| S\right|}\geqslant -1$. These eigenvalues are the same as those of $M$ since $M$ and $\Lambda$ represent the same linear transformation. Let $\pi_{\ast}\coloneqq\min_{1\leqslant i\leqslant\left| S\right|}\pi_{s_{i}}$. By standard results on mixing times in reversible Markov chains \citep[see][]{montenegro:FTTCS:2005}, we have
\begin{align}
	\psi\left(T\right) \leqslant \frac{1}{2}\sqrt{\frac{1-\pi_{\ast}}{\pi_{\ast}}}\max\left\{\left|\lambda_{2}\right| ,\left|\lambda_{\left| S\right|}\right|\right\}^{T} , \label{eq:mutationBoundReversible}
\end{align}

To see how the evolutionary process affects mixing along its lineages, we can again consider the case of a constant death rate for simplicity. If $d$ is the death rate for every individual, $i$, then the mutation process along any lineage is a Markov chain with transition matrix $M\left(d\right)\coloneqq\left(1-d\right) I+dM$. Moreover, $M\left(d\right)$ is reversible with respect to $\pi$ since $M$ is. Letting $\Lambda\left(d\right) =\left(1-d\right) I+d\Lambda$, we see that if $\Lambda\bm{v}=\lambda\bm{v}$, then $\Lambda\left(d\right)\bm{v}=\left(1-d\left(1-\lambda\right)\right)\bm{v}$. Since the map $\lambda\mapsto 1-d\left(1-\lambda\right)$ sends eigenvalues of $\Lambda$ to eigenvalues of $\Lambda\left(d\right)$, and since $\max_{2\leqslant i\leqslant\left| S\right|}\left| 1-d\left(1-\lambda_{i}\right)\right| =\max\left\{\left|1-d\left(1-\lambda_{2}\right)\right| ,\left|1-d\left(1-\lambda_{\left| S\right|}\right)\right|\right\}$, we have
\begin{align}
	\Psi\left(T\right) &\leqslant \frac{1}{2}\sqrt{\frac{1-\pi_{\ast}}{\pi_{\ast}}}\max\left\{\left|1-d\left(1-\lambda_{2}\right)\right| ,\left|1-d\left(1-\lambda_{\left| S\right|}\right)\right|\right\}^{T} .
\end{align}
Fig.~\ref{fig:upperBoundData} gives an example of this mixing time bound for the reversible chain with $\left| S\right| =3$ and transition matrix
\begin{align}
	M &= \begin{pmatrix}0.9429 & 0.0284 & 0.0287 \\ 0.0250 & 0.9638 & 0.0112 \\ 0.0360 & 0.0159 & 0.9481\end{pmatrix} . \label{eq:reversibleExample3}
\end{align}

\begin{figure}
	\centering
	\includegraphics[width=0.7\textwidth]{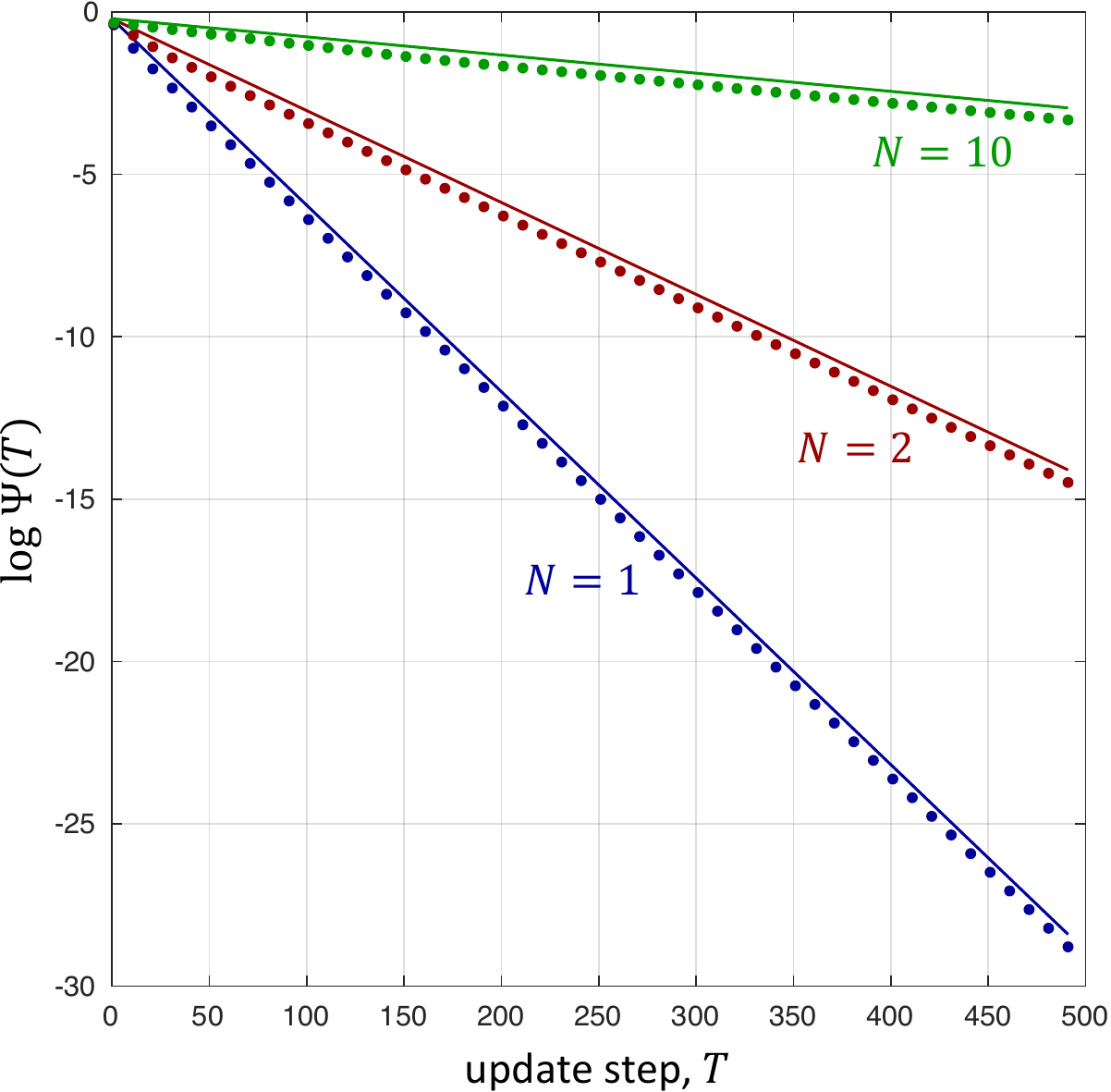}
	\caption{The distance between individual $i$'s trait and the stationary distribution after $T$ update steps, $\Psi\left(T\right)$, under death-birth updating. All of the marginal distributions converge to $\pi$ as $T\rightarrow\infty$, but the rates at which they do so depend on the population size, structure, and update rule. A population of size $N=1$ is the same as the underlying mutation process itself, and in this case $\Psi\left(T\right) =\psi\left(T\right)\leqslant C\alpha^{T}$ for some $C>0$ and $\alpha\in\left(0,1\right)$. The mutation process depicted here is reversible with $\left| S\right| =3$ traits and transition matrix Eq.~\ref{eq:reversibleExample3}, so we can set $C=\frac{1}{2}\sqrt{\frac{1-\pi_{\ast}}{\pi_{\ast}}}$ and $\alpha =\max\left\{\left|\lambda_{2}\right| ,\left|\lambda_{3}\right|\right\}$ (see Eq.~\ref{eq:mutationBoundReversible}). If the population has size $N$, then $\Psi\left(T\right)\leqslant C\left(1-\frac{1}{N}\left(1-\alpha\right)\right)^{T}$ since death rate for all individuals is $1/N$ under death-birth updating. Since $\Psi\left(T\right)$ converges rapidly to $0$ in all cases, we show both the predicted upper bounds (solid lines) and the actual values of $\Psi\left(T\right)$ (dots) on a logarithmic scale.\label{fig:upperBoundData}}
\end{figure}

On the other hand, if the death rate depends on $i$, we cannot necessarily transform the mutation process into one whose transition matrix is $M\left(d\right) =\left(1-d\right) I+dM$ for some $d$. Instead, starting at location $i$, one needs to keep track of where in the population each ancestor arises. Since $M$ is reversible, its time-reversal, $\widetilde{M}$, is again just $M$. Consider the Markov chain on $\left\{1,\dots ,N\right\}\times S$, $\left\{Y\left(T\right)\right\}_{T=0}^{\infty}$, with transitions given by
\begin{align}
	P_{\left(i,s\right) ,\left(j,s'\right)} &= 
	\begin{cases}
		e_{ji}M_{s,s'} & i\neq j\textrm{ or }s\neq s' , \\
		e_{ii}M_{s,s} + 1-d_{i} & i=j\textrm{ and }s=s' .
	\end{cases}
\end{align}
Starting at location $i$, one can then ask how many steps is required for the marginal trait distribution to be close to $\pi$ (in total variation). From the definition of $\Psi_{i}\left(T\right)$, it is easily seen that
\begin{align}
	\Psi_{i}\left(T\right) =\sup_{s\in S}\sup_{E\in\mathcal{F}\left(S\right)}\left|\mathbf{P}_{\left(i,s\right)}\left[Y\left(T\right)\in\left\{1,\dots ,N\right\}\times E\right] -\pi\left(E\right)\right| . \label{eq:DiReversed}
\end{align}
As a result, the mixing time of the marginal trait distribution of $\left\{Y\left(T\right)\right\}_{T=0}^{\infty}$ starting from $i$ is $\mathcal{T}_{i}\left(\varepsilon\right)$.

\subsection{Mutations induced by group actions}
One way to represent mutation is as a group action. Let $S$ be a finite set representing a collection of heritable traits, and let $G$ be a group (in the mathematical sense \citep[e.g., see][]{knapp:BB:2006}) representing the possible effects of mutation. Suppose that $G$ acts on $S$, meaning there exists a map $G\times S\rightarrow S$ sending $\left(g,s\right)$ to $gs\in S$ that satisfies the following properties: \textit{(i)} if $e$ is the identity element of $G$, then $es=s$ for every $s\in S$; and \textit{(ii)} if $g,h\in G$ and $s\in S$, then $\left(gh\right) s=g\left(hs\right)$. A special case is when $S=G$ and the action is the same as the binary operation in $G$.

For example, $G$ might be a group of three-dimensional rotations (a subgroup of $SO(3)$), acting on a set $S$ of possible orientations of some aspect of an organism's morphology. Or $G$ might be the group of permutations on $n$ elements, $\mathfrak{S}_{n}$, acting on the set $S$ of possible arrangements of $n$ transposable elements in a genome. An especially relevant example is $G=S=\left(\mathbb{Z}/2\mathbb{Z}\right)^{m}$; here $S$ can be understood as the group of binary strings (idealized genomes) of length $m$.

If $\nu\in\Delta\left(G\right)$ is a probability distribution on $G$, then one can define a Markov chain on $S$ with transitions
\begin{align}\label{eq:groupTransition}
	M_{s,s'} &\coloneqq \sum_{\substack{g\in G \\ gs=s'}} \nu\left(g\right) .
\end{align}
In order to ensure that this chain has a unique stationary distribution, we need to assume that the group action is transitive; that is, for any $s,s'\in S$, there exists $g\in G$ with $gs=s'$. For practical reasons, we make at least one of two other assumptions on the group action when considering the mixing times of $M$:

\subsubsection{Free group actions}
If the action of $G$ on $S$ is free, then $gs=hs$ implies that $g=h$. If the action of $G$ is on itself via its binary operation (in which case Eq.~\ref{eq:groupTransition} defines a random walk on $G$), then the action is free because every element of $G$ has an inverse. In other words, the only element of $G$ that sends $g\in G$ to itself is the identity element. A free action generalizes this property to actions of $G$ on another set, $S$.

For a free action, $M_{s,gs}\coloneqq\nu\left(g\right)$ for every $s\in S$ and $g\in G$. The distribution $\pi$ defined by $\pi_{g}=\left| G\right|^{-1}$ for $g\in G$ (i.e. the uniform distribution on $G$) is the stationary distribution for both $M$ and its time-reversal, $\widetilde{M}$, which satisfies $\widetilde{M}_{s,gs}=\nu\left(g^{-1}\right)$. A standard result on mixing times for random walks on groups \citep[see][Lemma 4.13 and Corollary 4.14]{levin:AMS:2008} is then easily seen to imply that $M$ and $\widetilde{M}$ have the same mixing time.

\subsubsection{Abelian group actions}
Let $G$ be an abelian group, meaning $gh=hg$ for every $g,h\in G$. An example of an abelian group is the set of binary strings of length $m$, $\left(\mathbb{Z}/2\mathbb{Z}\right)^{m}$, with the operation of addition. (The group of permutations on $n>2$ letters, $\mathfrak{S}_{n}$, is an example of a group that is \textit{not} abelian.) If $s\in S$ and $g,h\in G$ satisfy $gs=hs$, then $g^{-1}s=h^{-1}s$. The time-reversal of $M$ then satisfies
\begin{align}
	\widetilde{M}_{s,s'} &= \sum_{\substack{g\in G \\ gs=s'}} \nu\left(g^{-1}\right) .
\end{align}
Again, both $M$ and $\widetilde{M}$ have the uniform stationary distribution, and straightforward modifications of the arguments presented in \citep[][Lemma 4.13]{levin:AMS:2008} show that $M$ and $\widetilde{M}$ have the same mixing time.

If the group action is transitive and either free or abelian, it follows that one can obtain $\mathcal{T}_{i}\left(\varepsilon\right)$ from the time-reversed Markov chain on $\left\{1,\dots ,N\right\}\times S$, $\left\{Y\left(T\right)\right\}_{T=0}^{\infty}$, whose transitions are defined by
\begin{align}
	P_{\left(i,s\right) ,\left(j,s'\right)} &= 
	\begin{cases}
		e_{ji}\widetilde{M}_{s,s'} & i\neq j\textrm{ or }s\neq s' , \\
		e_{ii}\widetilde{M}_{s,s} + 1-d_{i} & i=j\textrm{ and }s=s' .
	\end{cases}
\end{align}
In particular, Eq.~\ref{eq:DiReversed} holds even though the mutation chain, $\left\{X\left(t\right)\right\}_{t=0}^{\infty}$, is not necessarily reversible.

\begin{remark}
	Although one can define a time-reversal of any mutation chain (not just one on a group or one that is reversible), the mixing time of the reversed chain need not coincide with the original chain. Therefore, the utility of reversing a chain and studying it along lineages is limited, as least for the purpose of analyzing mixing times.
\end{remark}

\section{Discussion}
The result on stationary trait frequencies is well-established in the special case of two competing types on a homogeneous population structure \citep{taylor:JTB:2007,debarre:JTB:2017}. Similarly, in population genetics, it has been noted that the so-called ``common ancestor" process is the same as the mutation process when there is no selection \citep{birky:PNAS:1988,fearnhead:JAP:2002,taylor:EJP:2007}. Comparisons between neutrally-evolving populations and their mutation processes have also been used to show that evolution favors types that are robust against mutation \citep{van_nimwegen:PNAS:1999}. Given the coupling of birth and mutation, how neutral evolution affects trait frequencies (allele frequencies, in the sexually-reproducing case) is a natural question to ask. Under mild assumptions, we have seen that neutral evolution results in the same trait frequencies as the mutation process.

This method can also be used to characterize convergence to the stationary distribution along the lineages. Since each lineage can contain at most one birth event per update step, the lineages mix at least as slowly as does the original mutation process. In general, if $\Psi_{i}\left(T\right)$ is the distance between the stationary distribution and the lineage leading to individual $i$, then $\Psi_{i}\left(T\right)\leqslant\mathbf{E}\left[\psi\left(B_{i}^{T}\right)\right]$, where $\psi\left(t\right)$ is the distance between the mutation process and the stationary distribution after $t$ steps and $B_{i}^{T}$ is a random variable giving the number of birth events along the lineage leading to $i$ at time $T$. The distribution of $B_{i}^{T}$ can be given explicitly in terms of the demographic variables of the update rule (Eq.~\ref{eq:generalBi}), which gives a bound on $\Psi_{i}\left(T\right)$ that is straightforward to calculate (although the expression itself is not necessarily succinct).

Matrix games, as well as the regeneration process, have finitely many possible types, but the main result on long-term trait frequency holds for mutation processes with continuous state spaces as well. A genetic type of this form could be an element of an interval such as $\left[0,1\right]$, for example, representing partial expression of a trait or the tendency to be of a particular binary type. Mutations can also be supported on uncountably infinitely many points in the trait space, one example being when an individual with type $x\in\left[0,1\right]$ mutates to a nearby type with probability determined by a Gaussian distribution centered at $x$ \citep{doebeli:S:2004,wakano:G:2012}. On a generic state space, the technical condition we require is that the mutation process be an ergodic Harris chain \citep{durrett:CUP:2009}.

Our focus here has been on evolving populations of fixed size and structure. These assumptions are not strictly necessary, but they do make the notation more convenient since an evolutionary process can then be described by a probability distribution over simple replacement events \citep{allen:JMB:2014}. Extinction also becomes a possibility when the population size fluctuates, which can make analyzing such a process more complicated \citep{haccou:CUP:2005,lambert:TPB:2006,faure:AAP:2014,hamza:JMB:2015,huang:PNAS:2015,mcavoy:TPB:2018}. Nonetheless, the intuition from our analysis here carries over to other situations: as long as an individual arises on a lineage with sufficiently many prior birth events, the mutation process along this lineage will converge to its stationary distribution, $\pi$. If all individuals in the population have this property sufficiently far into the future, then $\pi$ must give the long-term average trait frequencies in the population.

The study of trait dynamics along lineages leads to an interesting question: how does selection change the stationary trait distribution of individual $i$? This question, of course, is not new and has been studied in constant-fitness models over the past two decades via the so-called ``ancestral selection graph" \citep{krone:TPB:1997,neuhauser:TPB:1999,slade:G:2005}, which is a version of Kingman's coalescent \citep{kingman:JAP:1982} that allows for selective differences between the types. In populations with heterogeneous structure, selection can affect the marginal trait distributions in different ways at different locations, and there are still many open questions in this area, particularly with respect to frequency-dependent selection (games). For example, in an evolutionary game, the standard method of measuring the influence of selection is through its effects on average trait frequency in the population \citep{nowak:PTRSB:2009}. However, this type of averaging over the population does not provide a description of the locations at which a given trait is more likely to be found in a population, and it says nothing about selection's effects on mixing times.

Here, we have studied the individual lineages generated by neutral evolution. The trait distributions along those lineages are given by that of the mutation process. Therefore, the trait frequencies do not depend on population size, structure, and update rule. The mixing times of neutral evolution, however, do depend on these demographic components. Selection can perturb neutral evolutionary dynamics in complicated ways, and how it changes both marginal distributions and mixing times, especially within the realm of evolutionary game theory, is a more complicated question but nonetheless a natural extension of the one considered here.

\section*{Acknowledgments}
The authors are grateful to Christoph Hauert and John Wakeley for helpful discussions and to the referees for their comments on earlier drafts. A. M. and M. A. N. are supported by the Office of Naval Research, grant N00014-16-1-2914. B. Allen is supported by the National Science Foundation, award 1715315.

\end{document}